\documentclass[11pt]{article}
\usepackage{epsfig}
\usepackage{amssymb,amsthm}

\newcommand{\bilderpfad}{./}

\newtheorem{lemma}{Lemma}[section]
\newtheorem{theorem}[lemma]{Theorem}
\newtheorem{proposition}[lemma]{Proposition}

\renewcommand{\qed}{\vrule height .9ex width .8ex depth -.1ex}
\newcommand{\smallqed}{\quad
   \hbox{\vrule\vbox{\hrule\hbox to 4pt{%
         \vbox to 4pt{\hsize=4pt}}\hrule}\vrule}}

\renewenvironment{proof}{\noindent {\bf Proof.}}{\qed\medskip}

\newcommand{\df}{_{\rm def}}

\newcommand{\lcwd}{\mbox{\rm lcwd}}
\newcommand{\cwd}{\mbox{\rm cwd}}
\newcommand{\val}{{\rm val}}

\begin{document}

\title{Minimal forbidden induced subgraphs of graphs of
   bounded clique-width and bounded linear clique-width}

\author{
Daniel Meister\thanks{Theoretical Computer Science,
   University of Trier, Germany.
   Email: {\tt daniel.meister@uni-trier.de}}
\and
Udi Rotics\thanks{Netanya Academic College, Netanya, Israel.
   Email: {\tt rotics@netanya.ac.il}}}

\date{}

\maketitle

\begin{abstract}

In the study of full bubble model graphs of bounded clique-width
and bounded linear clique-width, we determined complete sets
of forbidden induced subgraphs, that are minimal in the class
of full bubble model graphs. In this note, we show that (almost
all of) these graphs are minimal in the class of all graphs. As
a corollary, we can give sets of minimal forbidden induced
subgraphs for graphs of bounded clique-width and for graphs of
bounded linear clique-width for arbitrary bounds.

\end{abstract}


\begin{section}{Preparation}
\label{section:introduction}

We consider graphs that are obtained from path powers. Path
powers are the powers of induced paths. They are proper
interval graphs. Let $k$ be an integer with $k\ge 1$, and let
$\Lambda= \langle x_1, \ldots, x_n\rangle$ be an ordering of
$n$ vertices. The {\it $k$-path power} with $k$-path
layout~$\Lambda$ is the graph on vertex set~$\{x_1, \ldots, x_n\}$,
and for $1\le i< j\le n$, $x_i$ and $x_j$ are adjacent if and
only if $j- i\le k$. It is an immediate consequence that the
1-path powers are exactly the induced paths. In this note, we
consider graphs that are obtained from combining path powers
into more complex graphs. We do not define and introduce the
used and necessary terminology. Instead, we refer to our main
paper, \cite{meisterRotics2013}.

We want to show upper bounds on the clique-width and linear
clique-width of some special graphs. Generally, upper bounds
can be shown explicitly, by constructing appropriate
clique-width expressions, or implicitly, by embedding as an
induced subgraph a graph into another graph of known bounded
clique-width or linear clique-width. We mainly apply the latter
approach. For two graphs~$G$ and $H$, we say that $H$ is
{\it embeddable into} $G$ if $H$ is isomorphic to an induced
subgraph of $G$. We consider only proper interval graphs here,
that can be represented by bubble models
\cite{heggernesMeisterPapadopoulos2009}, and embedding a proper
interval graph into a proper interval graph can be understood
as embedding a bubble model representation of the one graph
into a bubble model representation of the other graph. For
convenience, we may not distinguish between the graph itself
and a bubble model representation of the graph.

We will often embed into the following graph. Let $k$ be an
integer with $k\ge 2$. The graph~$J_k$ is a $k$-path power
on $(2k-1)(k+1)+1$ vertices and with $k$-path
layout~$\langle z_1, \ldots, z_m\rangle$. Let
$g=\df (k-1)(k+1)= k^2-1$. We will consider $J_k{-}z_g$. Note
here that $g$ depends on $k$, so that $z_g= z_{g_k}$ is a more
appropriate notation. For readability, we nevertheless write
$z_g$ instead of $z_{g_k}$. Examples of $J_k{-}z_g$ for three
values of $k$ are depicted in Figure~\ref{bild:examplesofJgraph},
where the graphs are represented by bubble models.

\begin{figure}[t]
\centering
\includegraphics[width=14cm]{\bilderpfad 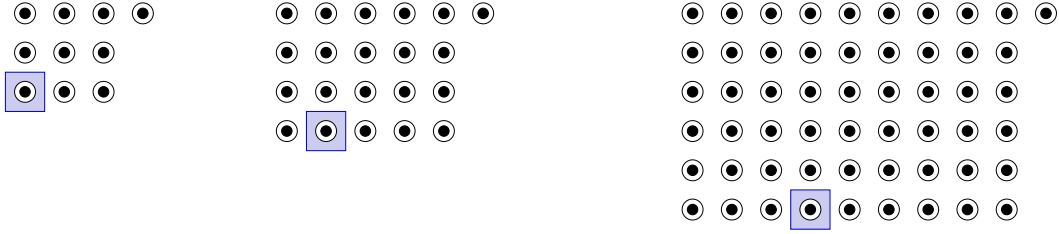}
\caption{Depicted are bubble model representations for $J_2$,
$J_3$ and $J_5$, seen from left to right. The highlighted vertex
is $z_g$.}
\label{bild:examplesofJgraph}
\end{figure}

\begin{lemma}
\label{lemma:upperboundforinducedsubgraphofJk}
For $k\ge 2$, $\lcwd(J_k{-}z_g)\le k+1$.
\end{lemma}

\begin{proof}
It suffices to observe that $J_k{-}z_g$ has an open $k$-model
as defined in \cite{meisterRotics2013}.
\end{proof}

\end{section}


\begin{section}{Induced subgraphs of $Z_k$}

Let $k$ be an integer with $k\ge 0$, and let
$n=\df k(k+1)+2$. The graph~$Z_k$ is a $k$-path power on
$n$ vertices and with $k$-path
layout~$\langle v_1, \ldots, v_n\rangle$.

\begin{theorem}[\cite{heggernesMeisterPapadopoulosRotics201X}]
For every $k\ge 0$, $\cwd(Z_k)\ge k+2$.
\end{theorem}

We show that $Z_k$ is a minimal graph of clique-width at least
$k+2$ and of linear clique-width at least $k+2$. Since
$\lcwd(G)\ge \cwd(G)$ for every graph~$G$, it suffices to
consider linear clique-width.

\begin{lemma}
\label{lemma:embeddingZkintoJk}
For every $k\ge 3$, every proper induced subgraph of $Z_k$
is an induced subgraph of $J_k{-}z_g$.
\end{lemma}

\begin{proof}
It suffices to show that $Z_k{-}v_t$ for every $1\le t\le k(k+1)+2$
is an induced subgraph of $J_k{-}z_g$. By an automorphism
argument, it suffices to restrict to $t$ with
$1\le t\le \frac{k(k+1)}{2}+1$. We show that $Z_k$ can be
embedded into $J_k$ such that $v_t$ is mapped to $z_g$.
Observe that $(k-1)(k+1)\ge \frac{k(k+1)}{2}+1$, so that
$t\le g$.

We define a mapping~$\varphi$ from $V(Z_k)$ into $V(J_k)$: let
$\varphi(v_i)=\df z_{g-t+i}$ for every $1\le i\le n$. We verify that
$\varphi$ has the desired properties. For every $1\le i\le n$,
\begin{eqnarray*}
1\ \le\ g-t+1\ \le\ g-t+i
&\le& g-t+n\\
&=& g-t+ k(k+1)+2\\
&\le& g-1+ k(k+1)+2\\
&=& (k-1)(k+1)+ k(k+1)+ 1\ =\ (2k-1)(k+1)+1\thinspace.
\end{eqnarray*}
Thus, $\varphi$ is a well-defined mapping from $V(Z_k)$ into
$V(J_k)$. It remains to see that $\varphi(v_t)= z_g$ is indeed
the case: $\varphi(v_t)= z_{g-t+t}= z_g$. Therefore, $\varphi$
defines an embedding of $Z_k{-}v_t$ into $J_k{-}z_g$ of the
desired form, and we conclude the claim of the lemma.
\end{proof}

\begin{proposition}
For every $k\ge 0$, every proper induced subgraph of $Z_k$
has linear clique-width at most $k+1$.
\end{proposition}

\begin{proof}
Let $H$ be a proper induced subgraph of $Z_k$. If $k\ge 3$
then $H$ is an induced subgraph of $J_k{-}z_g$ due to
Lemma~\ref{lemma:embeddingZkintoJk}, so that
$\lcwd(H)\le \lcwd(J_k{-}z_g)$, and thus,
$\lcwd(H)\le \lcwd(J_k{-}z_g)\le k+1$ due to
Lemma~\ref{lemma:upperboundforinducedsubgraphofJk}.

We consider the remaining cases for $k\le 2$. If $k= 0$ then
$Z_k$ has two vertices, and $H$ is a graph on at most one
vertex, and $\lcwd(H)\le 1$. If $k= 1$ then $Z_k$ is an
induced path on four vertices, and therefore, $\lcwd(H)\le 2$.
We consider the case of $k= 2$. Recall that $Z_2$ is a graph
on eight vertices, and following the proof of
Lemma~\ref{lemma:embeddingZkintoJk}, it suffices to consider
$Z_2{-}v_1$, $Z_2{-}v_2$, $Z_2{-}v_3$ and $Z_2{-}v_4$. Observe
that $g= 3$, and the mapping defined in
Lemma~\ref{lemma:embeddingZkintoJk} can be used to embed
$Z_2{-}v_t$ into $J_k{-}z_g$ for $1\le t\le 3$. The final
remaining case is $Z_2{-}v_4$. A linear 3-expression for
$Z_2{-}v_4$ can be obtained from the following vertex
ordering: $\langle v_8, v_7, v_6, v_5, v_3, v_2, v_1\rangle$.
Thus, $\lcwd(Z_2{-}v_4)\le 3$, and this completes the proof.
\end{proof}

\end{section}


\begin{section}{Induced subgraphs of $S^+_k$}

Let $k$ be an integer with $k\ge 2$, and let
$n=\df (k-1)(k+1)+2$.
\begin{itemize}
\item
   The graph~$S_k$ is obtained from a $k$-path power on
   $n$ vertices with $k$-path
   layout~$\langle v_1, \ldots, v_n\rangle$ by adding the
   vertices~$w_1, w_2, w_3, w_4$ and the
   edges~$w_1w_2, w_2v_1, v_nw_3, w_3w_4$.
\item
   For $k\ge 3$, the graph~$S^+_k$ is obtained from $S_k$ by
   adding the single vertex~$w^+$ of one of the following four
   neighbourhoods:
\begin{enumerate}
\item[a)]
   $N_{S^+_k}(w^+)= \{w_1, w_2, v_1, \ldots, v_k\}$
\item[b)]
   $N_{S^+_k}(w^+)= \{v_{n-k+1}, \ldots, v_n, w_3, w_4\}$
\item[c)]
   $N_{S^+_k}(w^+)= \{w_1, w_2, v_1, \ldots, v_{k-1}\}$
\item[d)]
   $N_{S^+_k}(w^+)= \{v_{n-k+2}, \ldots, v_n, w_3, w_4\}$.
\end{enumerate}
\end{itemize}
Observe that the four cases about the neighbourhood of $w^+$
generate two pairs of isomorphic graphs: $S^+_k$ in cases~a
and b are isomorphic, and $S^+_k$ in cases~c and d are
isomorphic.

\begin{theorem}[\cite{meisterRotics2013}]
\label{theorem:Spluskupperandlowerbounds}
$\hphantom{1}$
\begin{enumerate}
\item[1)]
   $\cwd(S^+_k)\ge k+2$ for $k\ge 3$
\item[2)]
   $\cwd(S_k)\le k+1< \lcwd(S_k)$ for $k\ge 3$
\item[3)]
   $\cwd(S_2)\ge 4$.
\end{enumerate}
\end{theorem}

We show that $S_k$ is a minimal graph of linear clique-width
at least $k+2$, and we show that $S^+_k$ with the case-c
neighbourhood of $w^+$ is a minimal graph of clique-width
at least $k+2$.

\begin{lemma}
\label{lemma:Splusklcwdupperbound}
For every $k\ge 3$ and $1\le t\le \frac{(k-1)(k+1)+1}{2}+1$,
$\lcwd(S^+_k{-}v_t)\le k+1$.
\end{lemma}

\begin{proof}
We follow the outline of the proof of
Lemma~\ref{lemma:embeddingZkintoJk} and show that $S^+_k{-}v_t$
is an induced subgraph of $J_k{-}z_g$, except for one particular
case.

We define a mapping from $V(S^+_k)$ into $V(J_k)$, that we define
in two steps. Let $\varphi(v_i)=\df z_{g-t+i}$ for every
$1\le i\le n$.
Then, for every $1\le i\le n$,
\begin{eqnarray*}
g- t+ i
&\le& g- 1+ n\\
&=& (k-1)(k+1)- 1+ (k-1)(k+1)+ 2\\
&=& (2k-2)(k+1)+ 1\ =\ (2k-1)(k+1)+1 - (k+1)\ =\
   m- (k+1)\thinspace.
\end{eqnarray*}
We want to extend $\varphi$ and map also $w_1, w_2, w_3, w_4$
and $w^+$. Observe that no vertex of $S^+_k$ has already
been mapped to $z_1, \ldots, z_{g-t}$ and $z_{g-t+n+1}, \ldots, z_m$.
Since $J_k$ is a $k$-path power and $\varphi(v_1)= z_{g-t+1}$,
we want to map $w_1$ to $z_{(g-t+1)-(k+1)}$, and analogously,
we want to map $w_2$ to $z_{g-t-k+1}$, and we want to map
$w_3$ and $w_4$ to $z_{g-t+n+k}$ and $z_{g-t+n+k+1}$. Furthermore,
if $w^+$ has the case-a or case-c neighbourhood then we want
to map $w^+$ to $z_{g-t-1}$ or $z_{g-t}$, and if $w^+$ has the
case-b or case-d neighbourhood then we want to map $w^+$ to
$z_{g-t+n+1}$ or $z_{g-t+n+2}$. If $g-t+1\ge k+2$ then
$(g-t+1)- (k+1)= g-t-k\ge 1$, and the extension for $w_1$ and
$w_2$ and $w^+$ is possible, and if $g-t+n+k+1\le m$ then the
extension for $w_3$ and $w_4$ and $w^+$ is possible. Note that
$g-t+n+k+1\le m$ is equivalent to $g-t+n\le m- (k+1)$, and
this condition is always satisfied, as we showed above.

\smallskip

Assume that $g-t+1\le k+1$, which means $g-t\le k$. Observe
the following:
$$k^2-1\ =\ g\ \le\ t+k\ \le\ \frac{(k-1)(k+1)+1}{2}+ 1+ k\
=\ \frac{1}{2}(k+1)^2+\frac{1}{2}\thinspace,$$
which is possible only for $k\le 3$, more precisely, for
$k= 3$ according to the assumptions of the lemma. Since
$g-t+1\le k+1$ is equivalent to $8+1-t\le 4$ in this case,
$t\ge 5$, and therefore, $t= 5$ must hold. This is the only
case for which the assumption about $g-t+1\le k+1$ is possible,
and $S^+_k{-}v_t$ is not embeddable into $J_k{-}z_g$.

To complete the proof also for this remaining case, we
construct a linear 4-expression for $S^+_3{-}v_5$. We consider
the case-a and case-c neighbourhood of $w^+$, and the two
other cases follow by isomorphy. We describe a linear
4-expression for $S^+_3{-}v_5$ where $w^+$ has the case-c
neighbourhood. It is an easy exercise, applying the deep
rectangle constructions in \cite{meisterRotics2013}, to
construct a linear 4-expression for the subgraph of $S^+_3$
induced by $\{v_2, v_3, v_4, v_6,\ldots, v_{10}, w_3, w_4\}$
such that $v_2$ has label~2, $v_3$ and $v_4$ have label~3
and the other vertices have label~1. The obtained linear
4-expression can be completed by adding the remaining
vertices in this order: $w^+, v_1, w_1, w_2$.

If $w^+$ has the case-a neighbourhood then $v_2$ and $v_3$
have label~2 and $v_4$ has label~3. We can conclude
$\lcwd(S^+_3{-}v_5)\le 4$.
\end{proof}

We emphasise that Lemma~\ref{lemma:Splusklcwdupperbound} is
true for all four neighbourhood cases of $w^+$ in $S^+_k$.

\newpage

\begin{proposition}
\label{proposition:properinducedSplussubgraphs}
Let $k\ge 3$. Let $S^+_k$ be obtained with the case-c
neighbourhood of $w^+$.
\begin{enumerate}
\item[1)]
   Every proper induced subgraph of $S^+_k$ has clique-width
   at most $k+1$.
\item[2)]
   Every proper induced subgraph of $S_k$ has linear clique-width
   at most $k+1$.
\item[3)]
   Every proper induced subgraph of $S_2$ has linear
   clique-width at most 3.
\end{enumerate}
\end{proposition}

\begin{proof}
The following is straightforward to see: $S_2{-}w_4$ is an
induced subgraph of $J_2{-}z_3$, and $S^+_k{-}w_4$ is an induced
subgraph of $J_k{-}z_g$. So, $\lcwd(S_2{-}w_4)\le 3$ and
$\lcwd(S^+_k{-}w_4)\le k+1$ due to
Lemma~\ref{lemma:upperboundforinducedsubgraphofJk}.
It is also easy to see that $S^+_k{-}w_3$ is the disjoint union
of $S^+_k\setminus \{w_3, w_4\}$ and $S^+_k[\{w_4\}]$, and $S_2{-}w_3$
is the disjoint union of $S_2\setminus \{w_3, w_4\}$ and
$S_2[\{w_4\}]$. In both cases, $w_4$ is an isolated vertex. We
directly conclude with the preceding result that
$\lcwd(S^+_k{-}w_3)\le k+1$ and $\lcwd(S_2{-}w_3)\le 3$. The
upper bounds analogously apply to $S^+_k{-}w_1$ and $S^+_k{-}w_2$
and $S_2{-}w_1$ and $S_2{-}w_2$. Finally, $\cwd(S^+_k{-}w^+)\le k+1$
due to the second statement of
Theorem~\ref{theorem:Spluskupperandlowerbounds}, since
$S^+_k{-}w^+= S_k$.

We consider induced subgraphs of $S^+_k$ that are obtained from
deleting a vertex~$v_t$, and we show $\lcwd(S^+_k{-}v_t)\le k+1$.
If $1\le t\le \frac{(k-1)(k+1)+1}{2}+1$ then
$\lcwd(S^+_k{-}v_t)\le k+1$ due to
Lemma~\ref{lemma:Splusklcwdupperbound}. If
$\frac{(k-1)(k+1)+1}{2}+1< t\le n$ then $S^+_k{-}v_t$ is isomorphic
to $S^+_k{-}v_{n-t+1}$ with the case-d neighbourhood of $w^+$,
and $\lcwd(S^+_k{-}v_t)\le k+1$ due to
Lemma~\ref{lemma:Splusklcwdupperbound}. Observe here
$$n-t+1\ <\
  (k-1)(k+1)+2-\left(\frac{(k-1)(k+1)+1}{2}+1\right)+1\
  =\ \frac{(k-1)(k+1)+1}{2}+1\thinspace,$$
so that Lemma~\ref{lemma:Splusklcwdupperbound} is indeed
applicable. This completes the proof about induced subgraphs
of $S^+_k$ for $k\ge 3$.

\smallskip

It remains to consider the remaining induced subgraphs of
$S_2$. Recall from the definition of $S_2$ that $S_2$ is obtained
from an induced path on $\{w_1, w_2, v_1, v_2, v_4, v_5, w_3, w_4\}$
by adding $v_3$ and making it adjacent to $v_1, v_2, v_4, v_5$.
We consider $S_2{-}x$ for $x\in \{v_1, \ldots, v_5\}$.

Clearly, $\lcwd(S_2{-}v_3)\le 3$. We consider $S_2{-}v_2$, which
is obtained from an induced path on
$\{w_1, w_2, v_1, v_3, v_5, w_3, w_4\}$ by adding $v_4$ and making
it adjacent to $v_3$ and $v_5$. It is straightforward to verify
that $S_2{-}v_2$ has a linear 3-expression, and thus,
$\lcwd(S_2{-}v_2)\le 3$. Analogously, $\lcwd(S_2{-}v_4)\le 3$.

We consider $S_2{-}v_1$, which is the disjoint union of
$S_2[\{v_2, v_3, v_4, v_5, w_3, w_4\}]$ and $S_2[\{w_1, w_2\}]$.
Since $S_2[\{v_2, v_3, v_4, v_5, w_3, w_4\}]$ is an induced
subgraph of $S_2{-}w_1$, the first paragraph of the proof
shows $\lcwd(S_2[\{v_2, v_3, v_4, v_5, w_3, w_4\}])\le 3$, so
that $S_2{-}v_1$ has a linear 3-expressions, and thus,
$\lcwd(S_2{-}v_1)\le 3$. Analogously, $\lcwd(S_2{-}v_5)\le 3$.
\end{proof}

Assume that $S^+_k$ is obtained with the case-a neighbourhood
of $w^+$. The proof of
Proposition~\ref{proposition:properinducedSplussubgraphs},
statement~1, is analogously applicable to all proper induced
subgraphs of $S^+_k$ but $S^+_k{-}w_4$: $S^+_k{-}w_4$ is not
an induced subgraph of $J_k{-}z_g$, since $w^+$ would be mapped
to $z_g$.

\end{section}


\begin{section}{Induced subgraphs of $M_{k, 1, l}$ and $M^{\pm}_2$}

Let $k$ be an integer with $k\ge 3$, and let $n=\df (k-1)(k+1)+1= k^2$.
The graph~$F_k$ is a $k$-path power on $n$ vertices and with
$k$-path layout~$\langle v_1, \ldots, v_n\rangle$, and the
graph~$F'_k$ is a $k$-path power on $n$ vertices and with
$k$-path layout~$\langle v'_1, \ldots, v'_n\rangle$. Let $l$
be an integer with $l\ge 0$. The graph~$M_{k, 1, l}$ is obtained
from the disjoint union of $F_k$ and $F'_k$ and $l$ new
vertices~$w_1, \ldots, w_l$ such that $\{v_n, w_1, \ldots, w_l, v'_n\}$
induces a 1-path power with 1-path
layout~$\langle v_n, w_1, \ldots, w_l, v'_n\rangle$.
Informally, $F_k$ and $F'_k$ are joined by an induced path of
length~$l+1$ that connects $v_n$ and $v'_n$. The special
graphs~$M^+_2$ and $M^-_2$ are depicted in
Figure~\ref{bild:smallestforbiddengraph}. If we do not need
to distinguish between $M^+_2$ and $M^-_2$, we shortly write
$M^{\pm}_2$.

\begin{figure}[t]
\centering
\includegraphics[width=14cm]{\bilderpfad 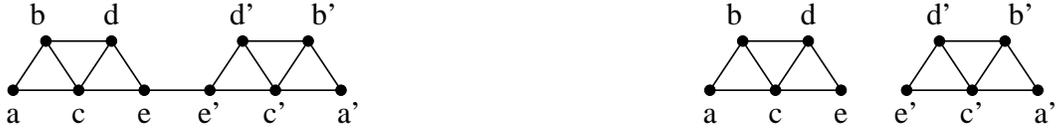}
\caption{Depicted are the graphs~$M^+_2$ to the left and $M^-_2$
to the right. Observe that $M^-_2$ is the disjoint union of two
gems.}
\label{bild:smallestforbiddengraph}
\end{figure}

\begin{theorem}[\cite{meisterRotics2013}]
\label{theorem:linearupperboundonMgraphs}
Let $k$ and $l$ be integers with $k\ge 3$ and $l\ge 0$.
\begin{enumerate}
\item[1)]
   $\lcwd(M^{\pm}_2)\ge 4$
\item[2)]
   $\lcwd(M_{k, 1, l})\ge k+2$.
\end{enumerate}
\end{theorem}

We show that $M^{\pm}_2$ is a minimal graph of linear clique-width
at least 4 and $M_{k, 1, l}$ is a minimal graph of linear
clique-width at least $k+2$.

\begin{proposition}
\label{proposition:inducedMsubgraphsandlinearcliquewidth}
For every $k\ge 3$ and $l\ge 0$,
\begin{enumerate}
\item[1)]
   every proper induced subgraph of $M^{\pm}_2$ has linear
   clique-width at most 3
\item[2)]
   every proper induced subgraph of $M_{k, 1, l}$ has linear
   clique-width at most $k+1$.
\end{enumerate}
\end{proposition}

\begin{proof}
We prove the first statement.

We consider $M^-_2$. Observe that
$M^-_2[\{a, b, c, d, e\}]$ has linear clique-width at most 3,
and every proper induced subgraph of $M^-_2[\{a, b, c, d, e\}]$
has a linear 3-expression with label~1 as an inactive label.
Thus, every proper induced subgraph of $M^-_2$ has linear
clique-width at most 3.

We consider $M^+_2$. Observe that $M^+_2{-}e$ and $M^+_2{-}e'$
are proper induced subgraphs of $M^-_2$, so that
$\lcwd(M^+_2{-}e)\le 3$ and $\lcwd(M^+_2{-}e')\le 3$ according
to the preceding paragraph. The other four remaining cases and
their automorphic equivalents are straightforward exercises.

\medskip

We prove the second statement. We show for every vertex~$x$
of $M_{k, 1, l}$ that $M_{k, 1, l}{-}x$ has linear clique-width
at most $k+1$. We distinguish between $x$ as a vertex from
$\{v_n, w_1, \ldots, w_l\}$ or from $\{v'_1\}$ or from
$\{v_2, \ldots, v_{n-1}\}$. The other, not considered cases
about $x$ directly follow by an automorphism argument. Let
$G=\df M_{k, 1, l}$.

We consider $x\in \{v_n, w_1, \ldots, w_l\}$. It is not
difficult to see that $F_k$ has a linear $(k+1)$-expression
with inactive label~1 that adds vertex~$v_n$ as the last vertex
with a unique label. The vertices~$w_1, \ldots, w_l$ can be
added by using two active labels only. So,
$G[\{v_1, \ldots, v_n, w_1, \ldots, w_l\}]$ has a linear
$(k+1)$-expression with inactive label~1, and so does
$G[\{w_1, \ldots, w_l, v'_n, \ldots, v'_1\}]$. As a
consequence, $G{-}x$ for $x\in \{v_n, w_1, \ldots, w_l\}$
has linear clique-width at most $k+1$.

We consider $x\in \{v'_1\}$. Then, $G{-}v'_1$ is an induced
subgraph of a graph with a short-end $k$-model, and
$\lcwd(G{-}v'_1)\le k+1$ due to the results from
\cite{meisterRotics2013}. The cases of even and odd parity
of $l$ need to be distinguished, and for the case of odd
parity, it is important to recall $k\ge 3$. Two examples,
for even and odd parity of $l$, are depicted and described
in Figure~\ref{bild:Mminusspecialembeddings}.

\begin{figure}[t]
\centering
\includegraphics[width=14cm]{\bilderpfad 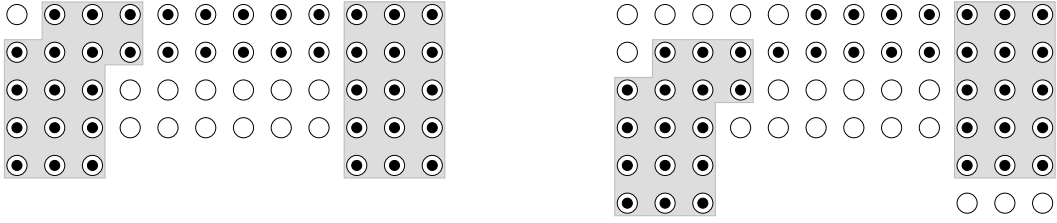}
\caption{Depicted are an embedding of a bubble model for
$M_{4, 1, 10}{-}v'_1$ into a short-end 4-model, to the left,
and an embedding of a bubble model for $M_{4, 1, 9}{-}v'_1$
into a short-end 4-model, to the right. The vertices of
$M_{4, 1, 10}{-}v'_1$ and $M_{4, 1, 9}{-}v'_1$ are highlighted
as the full vertices, and the subgraphs~$F_4$ and $F'_4{-}v'_1$
are marked by the shaded areas.}
\label{bild:Mminusspecialembeddings}
\end{figure}

We consider $x\in \{v_2, \ldots, v_{n-1}\}$. By the arguments
of the first case, we can assume a linear
$(k+1)$-expression~$\delta$ for
$G[\{v_n, w_1, \ldots, w_l, v'_n, \ldots, v'_1\}]$ such that
$v_n$ has label~$k$ and all other vertices have label~1 in
$\val(\delta)$. We want to extend $\delta$ into a linear
$(k+1)$-expression for $G{-}x$. As the main intermediate step,
we show how to construct a linear $(k+1)$-expression with
inactive label~1 for $F_k{-}x$. The ideas of the construction
resemble ideas for $J_k{-}z_g$ of Section~\ref{section:introduction}.
Let $x= v_p$. For an illustration, the four bubble models of
Figure~\ref{bild:Mconstructionexamples} show typical situations
about the deleted vertex~$x$. We distinguish between two cases
about the value of $p$ for the ease of description.
\begin{itemize}
\item
   Assume that $F_k[\{v_{p+1}, \ldots, v_n\}]$ has at most
   $(k-2)(k+1)$ vertices. This is the case for
   $k+2\le p\le n-1$.

   In this case, $F_k[\{v_{p+1}, \ldots, v_n\}]$ has a full
   bubble model that can be embedded into a deep rectangle of
   size~$k-2$. It is an exercise, by applying the construction
   ideas of Lemma~4.2 in \cite{meisterRotics2013}, to show a linear
   $(k+1)$-expression~$\beta$ with inactive label~1 for
   $F_k[\{v_{p+1}, \ldots, v_n\}]$ such that $v_n$ is inserted
   first and with label~$k$, and $v_{p+1}, \ldots, v_{p+k-1}$
   have label~$3, \ldots, k+1$ and all other vertices have
   label~1 in $\val(\beta)$.

   We remark that the linear $(k+1)$-expression of Lemma~4.2
   in \cite{meisterRotics2013} does not have an inactive label.
   Since we embed into a rectangle of size~$k-2$, we can
   nevertheless obtain a desired linear expression with inactive
   label~1.
\item
   Assume that $F_k[\{v_{p+1}, \ldots, v_n\}]$ has more than
   $(k-2)(k+1)$ vertices. This is the case for $2\le p\le k+1$.

   Assume $2\le p\le k$.
   In this case, $F_k[\{v_{p+1}, \ldots, v_n\}]$ has a full
   bubble model that can be embedded into a deep rectangle of
   size~$k-2$ and a shallow rectangle of size~1. The linear
   ${(k+1)}$-expression of Lemma~4.4 in \cite{meisterRotics2013}
   can be modified, by deleting some unnecessary vertices of
   the shallow rectangle, to obtain a linear
   $(k+1)$-expression~$\beta$ with inactive label~1 for
   $F_k[\{v_{p+1}, \ldots, v_n\}]$ that inserts $v_n$ first and
   with label~$k$, and $v_{p+1}, \ldots, v_{p+k-1}$ have
   label~$3, \ldots, k+1$ and all other vertices have label~1
   in $\val(\beta)$.

   Assume $p= k+1$.
   The construction of the preceding paragraph is not
   applicable in this case, since $v_n$ is the top vertex of
   the rightmost column, and the linear $(k+1)$-expression of
   Lemma~4.4 in \cite{meisterRotics2013} would insert $v_n$ late.
   Nevertheless, a desired expression exists and can be designed
   analogous to the expressions of the preceding case for $p\ge k+2$.
\end{itemize}
Combining the two linear $(k+1)$-expressions~$\delta$ and
$\beta$ yields a linear $(k+1)$-expression for
$G[\{v_{p+1}, \ldots, v_n, w_1, \ldots, w_l, v'_n, \ldots, v'_1\}]$.
The remaining vertices can be added according to the construction
of linear $(k+1)$-expressions for open $k$-models of
\cite{meisterRotics2013}.
\end{proof}

\begin{figure}[t]
\centering
\includegraphics[width=15cm]{\bilderpfad 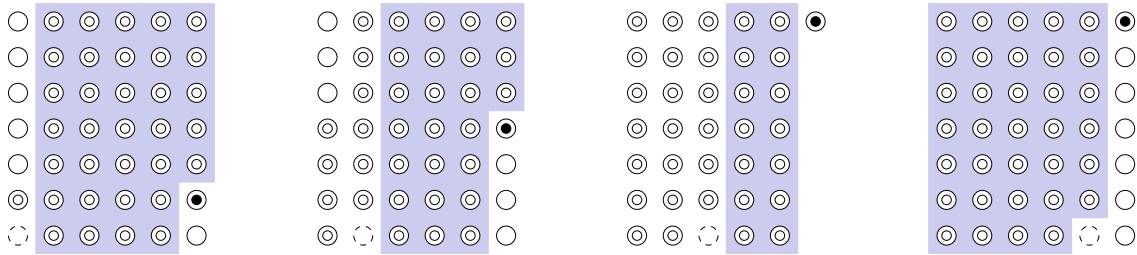}
\caption{The four bubble models show four situations about a
deleted vertex, namely $v_p$, of $F_k$, as they are considered
in the proof of
Proposition~\ref{proposition:inducedMsubgraphsandlinearcliquewidth}.
The figures are for the special case of $k= 6$. The vertices
of $F_k{-}v_p$ are marked as empty or full cycles in the
bubbles. The full cycle vertex is $v_n$. The dashed bubble
would contain $v_p$. The deleted vertex in the four figures is,
from left to right: $v_2$ and $v_{11}$ and $v_{21}$ and $v_{35}$.
The shaded area shows a deep rectangle and the neighbouring
column to the right. The full vertex is $v_n= v_{36}$, that
is included in the already constructed linear expression~$\delta$,
and the empty vertices are to be added.}
\label{bild:Mconstructionexamples}
\end{figure}

\end{section}



\begin{thebibliography}{10}


\bibitem{heggernesMeisterPapadopoulos2009}
P.~Heggernes, D.~Meister, C.~Papadopoulos.
\newblock A new representation of proper interval graphs with
   an application to clique-width.
\newblock {\em Electronic Notes in Discrete Mathematics},
   32:27--34, 2009.

\bibitem{heggernesMeisterPapadopoulosRotics201X}
P.~Heggernes, D.~Meister, C.~Papadopoulos, U.~Rotics.
\newblock Clique-width of path powers in linear time: a new
   characterisation of clique-width.
\newblock Submitted manuscript, 2012.

\bibitem{meisterRotics2013}
D.~Meister and U.~Rotics.
\newblock Clique-width of full bubble model graphs.
\newblock Submitted manuscript, 2013.\\[2pt]
   A full preliminary version is available as a technical report,
   no.~13-1, Technical Reports in Mathematics and Computer Science,
   University of Trier, 2013.


\end{thebibliography}
\end{document}